%% file: arxiv.tex
\documentclass{article}

\usepackage{arxiv}

\usepackage[english]{babel}
\usepackage{amsmath,amssymb,amsfonts,amssymb,amsthm}
\usepackage{algorithm,algpseudocode}
\usepackage{wrapfig, float}
\usepackage{paralist}
\usepackage{tikz-cd}

\usepackage{xcolor}
\usepackage{graphicx}
\usepackage{tikz-cd}
\usepackage{bm}

\usepackage[utf8]{inputenc} 
\usepackage[T1]{fontenc}    
\usepackage{hyperref}       
\usepackage{url}            
\usepackage{booktabs}       
\usepackage{amsfonts}       
\usepackage{nicefrac}       
\usepackage{microtype}      

\usepackage{graphicx}
\usepackage{doi}

\title{A Note on the Modeling Power of Different Graph Types}

\author{
\href{https://orcid.org/0000-0001-7166-9973}{Josephine Thomas}$^{a,1,}\thanks{Corresponding Author}$, \href{https://orcid.org/0000-0001-7912-0969}{Alice Moallemy-Oureh}$^{a,1}$, \href{https://orcid.org/0000-0002-2984-2119}{Silvia Beddar-Wiesing}$^{a,1}$, \href{https://orcid.org/0000-0002-1038-4242}{Rüdiger Nather}$^{a,1}$\\
${^a}$GAIN - Graphs in Artificial Intelligence and Neural Networks,\\ Faculty of Computer Science and Electrical Engineering,\\ 
University of Kassel, Wilhelmshöher Allee 71-73, 34121 Kassel, Germany\\
Email addresses: \{jthomas, amoallemy, s.beddarwiesing, r.nather\}@uni-kassel.de\\
\textbf{}\\
${^1}$Contributed equally\\
}

\renewcommand{\shorttitle}{A Note on the Modeling Power of Different Graph Types}

\hypersetup{
pdftitle={A Note on the Modeling Power of Different Graph Types},
pdfsubject={q-bio.NC, q-bio.QM},
pdfauthor={David S.~Hippocampus, Elias D.~Striatum},
pdfkeywords={graph expressivity, graph representation, information encoding},
}

\input{headerDefs}

\begin{document}
\maketitle

\begin{abstract}
Graphs can have different properties that lead to several graph types and may allow for a varying representation of diverse information. In order to clarify the modeling power of graphs, we introduce a partial order on the most common graph types based on an expressivity relation. The expressivity relation quantifies how many properties a graph type can encode compared to another type. Additionally, we show that all attributed graph types are equally expressive and have the same modeling power.
\end{abstract}

\keywords{graph expressivity $\And$ graph representation $\And$ information encoding}

\input{chapters/introduction}
\input{chapters/foundations}
\input{chapters/main}
\input{chapters/conclusion}

\input{chapters/acknowledgement}

\input{chapters/arxiv/bibliography}

\bibliographystyle{unsrtnat}







\end{document}

%% file: headerDefs.tex
\newtheorem{thm}{Theorem}[section]
\newtheorem{defn}[thm]{Definition}
\newtheorem{lem}[thm]{Lemma}

\newtheorem{rem}[thm]{Remark}
\newtheorem{exmp}[thm]{Example}
\newtheorem*{thm*}{Theorem}
\newtheorem*{lem*}{Lemma}

\newtheoremstyle{rmrk}{3pt}{3pt}{}{0pt}{\itshape}{:}{.5em}{}
\theoremstyle{rmrk}
\newtheorem*{rem*}{\underline{Remark}}

\newcommand{\hypset}[1]{\ensuremath{\mathsf{#1}}}
\newcommand{\abb}{\ensuremath{\mathbin{:}}}

\newcommand{\set}[1]{\mathcal{#1}}

%% file: chapters/introduction.tex
\section{Introduction}
Let $g=(\set{V},\set{E})$ be a simple graph. The attributed version $g'=(\set{V},\set{E},\alpha,\omega)$ with attribute mappings $\alpha:\set{V}\rightarrow \set{A}$ and $\omega:\set{E}\rightarrow\set{B}$ into attribute spaces $\set{A},\set{B}$ may intuitively express more information than the unattributed graph. However, there are several different graph types, for example, heterogeneous graphs, multigraphs, hypergraphs, directed or undirected graphs, and the question arises: \textit{which graph type can express how much information? }

Irrespective of the graph type, two graphs $g_1 = (\set{V}_1, \set{E}_1), g_2 = (\set{V}_2, \set{E}_2)$ were compared in standard literature by checking if they are structurally isomorphic \cite{jour_grohe_2020}, i.e., finding bijections $f_1: \set{V}_1 \rightarrow \set{V}_2, f_2: \set{E}_1 \rightarrow \set{E}_2$ between their node and edge sets. For the attributed case, the equality of two graphs requires additional bijections between the attribute sets of nodes and edges, respectively. 
In this note, however, two graphs are not distinguished by their structure, but the modeling power of different graph types are analyzed, i.e., the number of graph properties that can be encoded by sets of graphs that differ in their structure by definition. We call the relation \textit{graph type expressivity} and define it as follows.

\begin{defn}[Graph Type Expressivity]
  \label{defn_expressivity}
  A graph type $\set{G}_2$ is \textbf{at least as {expressive}} as a graph type $\set{G}_1$, if and only if $\set{G}_2$ encodes at least as many graph properties as $\set{G}_1$ denoted as $\set{G}_1\preccurlyeq \set{G}_2$. In case both types encode the same graph properties, i.e., $\set{G}_1\preccurlyeq \set{G}_2 \wedge \set{G}_2\preccurlyeq \set{G}_1$, it is denoted as $\set{G}_1 \approx \set{G}_2$.
\end{defn}
The fact that the graph type expressivity induces a partial order leads to the following theorems for unattributed graphs \ref{theorem_unattributed}, partially attributed graphs \ref{theorem_partially_attributed} and attributed graphs \ref{theorem:attributesGraphTYpesEquallyExpressive}.

\begin{thm}\label{theorem_unattributed}
    Let $\set{G}_{h}, \set{G}_{\bar{h}}, \set{G}_{d}$ and $\set{G}_{het}$ be the sets of hypergraphs, non-hypergraphs, directed and heterogeneous graphs. Then it holds $\set{G}_{h}\approx \set{G}_{\bar{h}}$, and $\set{G}_d \preccurlyeq \set{G}_{het}$.
\end{thm}

\begin{thm}\label{theorem_partially_attributed}
Let $\set{G}_{a}, \set{G}_{ia}$ be the sets of attributed and integer-attributed graphs, and $\set{G}_{\bar{a}}, \set{G}_{\bar{a}, het}, \set{G}_{\bar{a}, m}$ be the sets of unattributed graphs, unattributed heterogeneous and unattributed multigraphs. Then it holds $\set{G}_{\bar{a}} \preccurlyeq \set{G}_a,\ \set{G}_{a} \approx \set{G}_{\bar{a}, het}$, and $\set{G}_{ia} \approx \set{G}_{\bar{a}, m}$.
\end{thm}

\begin{thm}\label{theorem:attributesGraphTYpesEquallyExpressive}
  All attributed graph types are equally expressive.
\end{thm}

The paper is structured as follows. In sec.~\ref{sec_foundations} the different graph types and properties are defined. The proof of thm.~\ref{theorem_unattributed} is given in sec.~\ref{section_proofs_unattributed}, the proof of \ref{theorem_partially_attributed} can be found in sec.~\ref{section_proofs_partially_attributed} and thm.~\ref{theorem:attributesGraphTYpesEquallyExpressive} is proven in sec.~\ref{section_proofs_attributed}. For the composition of graph types, a note is given in sec.~\ref{sec_expressivity_composed_graphs} . Finally, the paper is concluded in sec.~\ref{sec:conclusion}.

%% file: chapters/foundations.tex
\section{Preliminaries}
\label{sec_foundations}

\begin{defn}[Elementary Graphs]\label{defn:StaticGraphsElm}
  \leavevmode
  \begin{compactenum}
  \item A \textbf{directed graph} is a tuple $G=(\set{V}, \set{E})$ with $\set{V}\subset\mathbb{N}$ and $\set{E}\subseteq\set{V}\times\set{V}$. The set of all directed graphs is denoted here as $\set{G}_d$.
  \item A \textbf{(generalized) directed hypergraph} is defined as $G=(\set{V},\set{E})$ with nodes $\set{V}\subset\mathbb{N}$ and hyperedges \mbox{$\set{E}\subseteq \{(\hypset{x},f_i)_i\mid \hypset{x}\subseteq\set{V},\, f_i\abb\hypset{x}\rightarrow\mathbb{N}_0\}$} including a numbering of the nodes in the hyperedge.\footnote{W.l.o.g.~it can be assumed that the numbering is gap-free, so if there exists a node $u\in\hypset{x}$ with ${f(u)=k>1}$ then there will also exist a node $v$ s.t.~$f(v)=k-1$.} 
  The set of all (generalized) directed hypergraphs is denoted as $\set{G}_h$.
  \end{compactenum}
\end{defn}

\begin{figure}[H]
\begin{minipage}{.45\linewidth}
\begin{rem}
  Def.~\ref{defn:StaticGraphsElm}[2] differs from the common definition of a directed hypergraph with edges ${\set{E}\subseteq\{(\hypset{x},\hypset{y})\mid \hypset{x},\hypset{y}\subseteq \set{V}\}}$ \cite{book_bretto_2013}. We obtained a more generalized definition by introducing a numbering mapping for each hyperedge that indicates an ordering of the nodes in a hyperedge. On a \textbf{generalized hyperedge} we defined the ordering by equipping the node set $\hypset{x}$ with a function $f:\hypset{x}\rightarrow\mathbb{N}$ s.t.~for ${u,v\in\hypset{x} : u \leq v \Leftrightarrow f(u)\leq f(v)}$. With this, the common notion of a directed hyperedge $(\hypset{x},\hypset{y})$ can be depicted by mapping the nodes in $\hypset{x}$ to $1$ and the nodes in $\hypset{y}$ to $2$.
  \end{rem}
\end{minipage}\hfill
\begin{minipage}{.5\linewidth}
 \centering\scriptsize
 \def\svgwidth{\textwidth}
  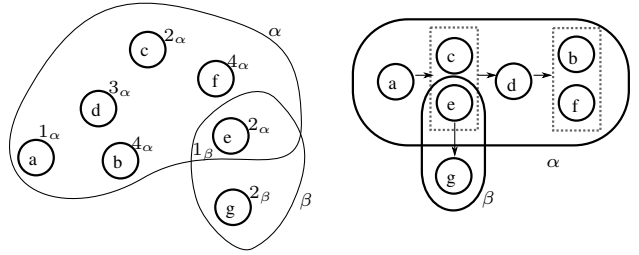
  \caption[hypIntro]{Left: Two generalized hyperedges with numbering. Right: The same hyperedges with arrows indicating the ordering of the nodes.}
\end{minipage}
\end{figure}


\begin{defn}[Graph Properties]\label{defn:StaticStructuralProp}
\leavevmode
  An elementary graph $G = (\set{V}, \set{E})$ is called
  \begin{compactenum}
  \item \textbf{undirected} if the directions of the edges are irrelevant, i.e.,
    \begin{compactenum}
    \item for directed graphs: if $(u,v) \in \set{E}$ whenever $(v,u)\in\set{E}$ for $u,v\in\set{V}$.
    \item for directed hypergraphs: if $f_i\abb\hypset{x}\rightarrow 0$ for all $(\hypset{x},f_i)_i\in\set{E}$. Abbreviated by $\set{E}\subseteq\{ \hypset{x}\mid \hypset{x}\subseteq\set{V}\}$.
    \end{compactenum}
    In what follows $\set{G}_{u}$ is the set of all undirected graphs.
  \item \textbf{multigraph} if it is a multi-edge graph, i.e., the edges $\set{E}$ are defined as an (ordered) multiset, denoted by $\{\!|\cdot |\!\}$, a multi-node graph, i.e., the node set $\set{V}$ is an (ordered) multiset, or both. All multigraphs are written as the set $\set{G}_m$.
  \item \textbf{heterogeneous} if the nodes or edges can have different types (node- or edge-heterogeneous). I.e., the node set is determined by $\set{V}\subseteq\mathbb{N}\times \set{S}$ with a node type set $\set{S}$ and the edges are extended by the edge type set $\set{R}$. The set of all heterogeneous graphs is denoted as $\set{G}_{het}$.
  \item \textbf{attributed} if the nodes $\set{V}$ or edges $\set{E}$ are equipped with node- or edge attributes. These attributes are given by a node and edge attribute function, respectively, i.e., $\alpha:\set{V}\rightarrow\set{A}$ and $\omega:\set{E}\rightarrow \set{W}$, where $\set{A}$ and $\set{W}$ are arbitrary attribute sets. 
    The set of all attributed graphs is denoted as $\set{G}_a$.
    \item \textbf{dynamic} if the graph structure or the graph properties are time dependent. In the following, the notion $G_i = (\set{V}_i, \set{E}_i), \;t_i \in T$ is used, where $T$ is a finite set of (not necessarily equidistant) timestamps to emphasize the time-dependence and therefore the dynamics.
  \end{compactenum}
\end{defn}

In the literature there are many graph types composed from the above definitions of graph properties, but the resulting graph types are not always named accordingly. Therefore, in the following important composed graphs are listed and named according to the definitions from above.
\begin{exmp}[Composed Static Graphs]\label{defn:StaticGraphsComb}
\leavevmode
  \begin{compactenum}
  \item \textbf{Knowledge graphs} are defined in several ways. In \cite{jour_wang_2021}, they are defined as heterogeneous directed graphs, while in \cite{inproc_zhu_2020} knowledge graphs are heterogeneous graphs. But some definitions that do not define knowledge graphs as a graph consisting of a combination from the aforementioned types, see  \cite{inproc_ehrlinger_2016}.
  \item A \textbf{multi-relational graph} \cite{jour_hamilton_2020} is an edge-heterogeneous but node-homo-\ geneous graph.
  \item A \textbf{content-associated heterogeneous graph} is a heterogeneous graph with node attributes that correspond to  heterogeneous data such as, e.g., attributes, text or images \cite{inproc_zhang_2019}. 
  \item A \textbf{multiplex graph}/\textbf{multi-channel graph} corresponds to an edge-heterogeneous graph with self-loops \cite{jour_hamilton_2020}. Here, we have $k$ layers, where each layer consists of the same node set $\set{V}$, but different edge sets $\set{E}^{(k)}$. Additionally, inter-layer edges $\tilde{\set{E}}$ exist between the same nodes across different layers.
  \item A \textbf{multiscale graph} is a multiplex graph with edges between different nodes of differing layers \cite{inproc_li_2020}.
  \item A \textbf{spatio-temporal graph} is a multiplex graph where edges per each layer are interpreted as spatial edges and the inter-layer edges indicate temporal steps between a layer at time step $t$ and $t+1$. They are called temporal edges \cite{arx_kapoor_2020}.	
  \end{compactenum}
\end{exmp}


\begin{rem}
  Graphs can have further additional semantic properties such as connectedness, (a-)cyclicity, scale-freeness, etc. However, these are independent from the basic graph type. Since we consider graphs syntactically in this article, we do not go into further graph properties that result from interpretations of the given graph structure. For further readings see \cite{book_korte_2012}.
\end{rem}


%% file: graphics/hyperIntro.pdf_tex
\begingroup%
  \makeatletter%
  \providecommand\color[2][]{%
    \errmessage{(Inkscape) Color is used for the text in Inkscape, but the package 'color.sty' is not loaded}%
    \renewcommand\color[2][]{}%
  }%
  \providecommand\transparent[1]{%
    \errmessage{(Inkscape) Transparency is used (non-zero) for the text in Inkscape, but the package 'transparent.sty' is not loaded}%
    \renewcommand\transparent[1]{}%
  }%
  \providecommand\rotatebox[2]{#2}%
  \ifx\svgwidth\undefined%
    \setlength{\unitlength}{446.49116036bp}%
    \ifx\svgscale\undefined%
      \relax%
    \else%
      \setlength{\unitlength}{\unitlength * \real{\svgscale}}%
    \fi%
  \else%
    \setlength{\unitlength}{\svgwidth}%
  \fi%
  \global\let\svgwidth\undefined%
  \global\let\svgscale\undefined%
  \makeatother%
  \begin{picture}(1,0.39917389)%
    \put(0.12883797,0.18610303){\color[rgb]{0,0,0}\makebox(0,0)[lb]{\smash{}}}%
    \put(0.04715246,0.18297229){\color[rgb]{0,0,0}\makebox(0,0)[lb]{\smash{$1_\alpha$}}}%
    \put(0.24943556,0.33793311){\color[rgb]{0,0,0}\makebox(0,0)[lb]{\smash{$2_\alpha$}}}%
    \put(0.38477568,0.19410396){\color[rgb]{0,0,0}\makebox(0,0)[lb]{\smash{$2_\alpha$}}}%
    \put(0.16025356,0.2537591){\color[rgb]{0,0,0}\makebox(0,0)[lb]{\smash{$3_\alpha$}}}%
    \put(0.35113395,0.29622785){\color[rgb]{0,0,0}\makebox(0,0)[lb]{\smash{$4_\alpha$}}}%
    \put(0.19703399,0.16382031){\color[rgb]{0,0,0}\makebox(0,0)[lb]{\smash{$4_\alpha$}}}%
    \put(0,0){\includegraphics[width=\unitlength,page=1]{graphics/hyperIntro.pdf}}%
    \put(0.03199947,0.13826623){\color[rgb]{0,0,0}\makebox(0,0)[lb]{\smash{a}}}%
    \put(0,0){\includegraphics[width=\unitlength,page=2]{graphics/hyperIntro.pdf}}%
    \put(0.13259456,0.21724864){\color[rgb]{0,0,0}\makebox(0,0)[lb]{\smash{d}}}%
    \put(0,0){\includegraphics[width=\unitlength,page=3]{graphics/hyperIntro.pdf}}%
    \put(0.16801087,0.1332839){\color[rgb]{0,0,0}\makebox(0,0)[lb]{\smash{b}}}%
    \put(0,0){\includegraphics[width=\unitlength,page=4]{graphics/hyperIntro.pdf}}%
    \put(0.21185233,0.31275621){\color[rgb]{0,0,0}\makebox(0,0)[lb]{\smash{c}}}%
    \put(0,0){\includegraphics[width=\unitlength,page=5]{graphics/hyperIntro.pdf}}%
    \put(0.32729478,0.2624744){\color[rgb]{0,0,0}\makebox(0,0)[lb]{\smash{f}}}%
    \put(0,0){\includegraphics[width=\unitlength,page=6]{graphics/hyperIntro.pdf}}%
    \put(0.34567827,0.17318274){\color[rgb]{0,0,0}\makebox(0,0)[lb]{\smash{e}}}%
    \put(0,0){\includegraphics[width=\unitlength,page=7]{graphics/hyperIntro.pdf}}%
    \put(0.61119012,0.26055344){\color[rgb]{0,0,0}\makebox(0,0)[lb]{\smash{a}}}%
    \put(0,0){\includegraphics[width=\unitlength,page=8]{graphics/hyperIntro.pdf}}%
    \put(0.90125531,0.30090906){\color[rgb]{0,0,0}\makebox(0,0)[lb]{\smash{b}}}%
    \put(0,0){\includegraphics[width=\unitlength,page=9]{graphics/hyperIntro.pdf}}%
    \put(0.90719735,0.22342922){\color[rgb]{0,0,0}\makebox(0,0)[lb]{\smash{f}}}%
    \put(0,0){\includegraphics[width=\unitlength,page=10]{graphics/hyperIntro.pdf}}%
    \put(0.70570471,0.3028113){\color[rgb]{0,0,0}\makebox(0,0)[lb]{\smash{c}}}%
    \put(0,0){\includegraphics[width=\unitlength,page=11]{graphics/hyperIntro.pdf}}%
    \put(0.70585278,0.22646938){\color[rgb]{0,0,0}\makebox(0,0)[lb]{\smash{e}}}%
    \put(0,0){\includegraphics[width=\unitlength,page=12]{graphics/hyperIntro.pdf}}%
    \put(0.80205284,0.25752816){\color[rgb]{0,0,0}\makebox(0,0)[lb]{\smash{d}}}%
    \put(0,0){\includegraphics[width=\unitlength,page=13]{graphics/hyperIntro.pdf}}%
    \put(0.38674948,0.08378673){\color[rgb]{0,0,0}\makebox(0,0)[lb]{\smash{$2_\beta$}}}%
    \put(0,0){\includegraphics[width=\unitlength,page=14]{graphics/hyperIntro.pdf}}%
    \put(0.34916621,0.05860983){\color[rgb]{0,0,0}\makebox(0,0)[lb]{\smash{g}}}%
    \put(0.41716981,0.34565999){\color[rgb]{0,0,0}\makebox(0,0)[lb]{\smash{$\alpha$}}}%
    \put(0.46770041,0.07157735){\color[rgb]{0,0,0}\makebox(0,0)[lb]{\smash{$\beta$}}}%
    \put(0.29515833,0.15364668){\color[rgb]{0,0,0}\makebox(0,0)[lb]{\smash{$1_\beta$}}}%
    \put(0,0){\includegraphics[width=\unitlength,page=15]{graphics/hyperIntro.pdf}}%
    \put(0.70459705,0.10917951){\color[rgb]{0,0,0}\makebox(0,0)[lb]{\smash{g}}}%
    \put(0.76121039,0.07487386){\color[rgb]{0,0,0}\makebox(0,0)[lb]{\smash{$\beta$}}}%
    \put(0.86562564,0.13369748){\color[rgb]{0,0,0}\makebox(0,0)[lb]{\smash{$\alpha$}}}%
    \put(0,0){\includegraphics[width=\unitlength,page=16]{graphics/hyperIntro.pdf}}%
  \end{picture}%
\endgroup%

%% file: chapters/main.tex
\section{Modeling Power of Unattributed Graphs}\label{section_proofs_unattributed}

\begin{lem}\label{eq:hyper}
  The set of directed hypergraphs and non-hypergraphs are equally expressive, i.e.,
    $\set{G}_{h}\approx \set{G}_{\overline{h}}$.
  
  \begin{proof}
    Let $g\in\set{G}_{h}$, $g'\in\set{G}_{\bar{h}}$ be an arbitrary hyper- and non-hypergraph of the forms 
    $g = (\set{V}, \set{E})$ with  $\set{V}\subset\mathbb{N},\; {\set{E}\subseteq \{(\hypset{x},f_i)_i\mid \hypset{x}\subseteq\set{V},\, f_i\abb\hypset{x}\rightarrow\mathbb{N}_0\}}$ and $g'=(\set{V}',\set{E}')$ with $\set{V}'\subset \mathbb{N}$ and $\set{E}'\subseteq \set{V}'\times \set{V}'.$


    
    \begin{compactenum}
    \item[$\preccurlyeq$:] Let $f\abb \set{G}_h\rightarrow \text{Im}(f)\subseteq\set{G}_{\bar{h}}$ be a mapping with $f(g)=g'=(\set{V}, \set{E}')$ and $\set{E}' = \set{E}^\star \cup \set{E}^{\star\star}$, with
      \begin{compactenum}
      \item the translation of undirected hyperedges to fully connected subgraphs (cliques):\\ ${\set{E}^\star = \{(u,v)\mid \exists (\hypset{x},f_i)_i\in \set{E}\abb f_i(w)=0\,\forall w\in\hypset{x}\wedge u,v\in \hypset{x}\}}$ as illustrated in fig.~\ref{fig:hyperedges} a),
      \item the translation of directed hyperedges to chains of fully connected bipartite subgraphs (bicliques):\\ ${\set{E}^{\star\star} = \{(u,v)\mid \exists(\hypset{x}, f_i)_i\in \set{E}\abb u,v\in\hypset{x}\wedge f_i(v)=f_i(u)+1 \}}$ as illustrated in fig.~\ref{fig:hyperedges} b).
      \end{compactenum}
      Then $f$ is bijective, while the injectivity is insured by the ordering functions $f_i$ of the corresponding edges.
    
    \item[$\succcurlyeq$:] This direction is straight forward since every non-hypergraph is per definition a hypergraph.
    
    \end{compactenum}

    \begin{figure}[h!]
      \centering
      \includegraphics{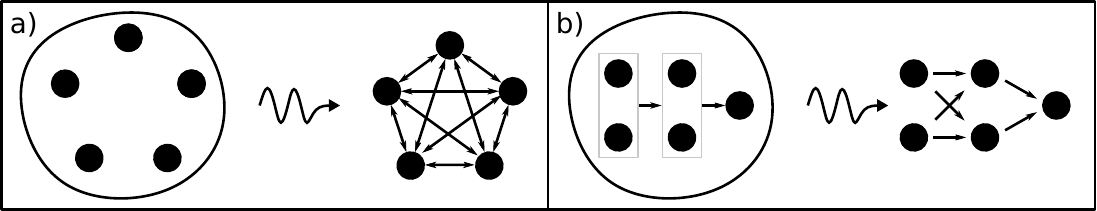}
      \caption{embedding of hyperedges to directed edges: a) undirected hyperedge, b) directed hyperedge.}
      \label{fig:hyperedges}
    \end{figure}
  \end{proof}
\end{lem}

\begin{lem}\label{lem_directed_heterogeneous}
Let $\set{G}_d$ be the set of directed graphs and $\set{G}_{het}$ the set of undirected heterogeneous graphs. Then $\set{G}_d \preccurlyeq \set{G}_{het}$.
\end{lem}

\begin{proof}
Let $g := (\set{V},\set{E}) \in\set{G}_d$ be a directed graph with $\set{V}\subset\mathbb{N},\ \set{E}\subseteq\set{V}\times\set{V}$. 
Then the function ${f: \set{G}_d\rightarrow \text{Im}(f)\subseteq \set{G}_{het}}$ defined by $ f(g) = g' := (\set{V}, \set{E}')$ with $\set{E}' = \{(u,v,0), (v, u, 0)\mid (u,v)\in \set{E}\}\cup \{(u,v,1), (v, u, 1)\mid (v,u)\in \set{E}\}$ is a bijective embedding.
\end{proof}

From lem.~\ref{eq:hyper} and lem.~\ref{lem_directed_heterogeneous} follows directly thm.~\ref{theorem_unattributed}.

\section{Modeling Power between Attributed and Unattributed Graphs}\label{section_proofs_partially_attributed}

\begin{lem}\label{lemma:nonattributet_2_attributed}
  Let $G_{a}$ be the sets of attributed graphs and $\set{G}_{\bar{a}}$ the set of unattributed graphs. Then it holds
    $\set{G}_{\bar{a}}\preccurlyeq\set{G}_{a}$.
  \begin{proof}
    Let $g\in\set{G}_{\bar{a}}$ be an arbitrary graph of the form 
    $g = (\set{V}, \set{E}) \text{ with }\set{V}\subset\mathbb{N},\; \set{E}\subseteq\set{V}\times\set{V}$. Then the function $f\abb \set{G}_{\bar{a}} \longrightarrow \set{G}_a$, defined by 
    \begin{align*}
      (\set{V}, \set{E}) &\longmapsto (\set{V}, \set{E}, \alpha, \omega), \; \text{ with } \; 
                           \alpha\abb\set{V} \longrightarrow \set{A} \cup \{\sigma\}, \; \omega\abb\set{E} \longrightarrow \set{B} \cup \{\sigma\},\\
      \alpha(v) &= \sigma \; \forall \; v \in \set{V} \; \text{ and } \;
                  \omega(e) = \sigma \; \forall \; e \in \set{E}.
    \end{align*}
    where $\sigma$ is the blank symbol, is a bijective embedding from the unattributed graph into the attributed graph type.
  \end{proof}
\end{lem}

The inverse direction of lem.~\ref{lemma:nonattributet_2_attributed} must not necessarily hold. The main problem lies in encoding the node and edge attributes into the graphs structure. However, for some special cases we show in the following (cf.~lem.~\ref{lem_attributed_heterogen}, \ref{lem_multigraph_unattrib}), that restricted versions of the inverse direction hold. These could be seen as indicators for the reverse direction hypothesis of lem.~\ref{lemma:nonattributet_2_attributed}.

\begin{lem}\label{lem_attributed_heterogen}An attributed graph can be embedded into an unattributed heterogeneous graph type, i.e., it holds
  $\set{G}_a \preccurlyeq \set{G}_{\bar{a},het}$.
  \begin{proof}
    Let $g\in\set{G}_a$ be an arbitrary graph with
    $g = (\set{V}, \set{E}, \alpha, \omega)$, $\set{V}\subset\mathbb{N},\ \set{E}\subseteq\set{V}\times\set{V},\ \alpha\abb \set{V}\rightarrow \mathcal{C}$ and $\omega\abb\set{E}\rightarrow\set{C}$.
    Then ${f\abb \set{G}_a \rightarrow \text{Im}(f) \subseteq \set{G}_{\bar{a},het}}$ with $f(g) = g' := (\set{V}', \set{E}')$, where $\set{V}' = \set{V} \times \alpha(\set{V})$
    $\set{E}' = \set{E} \times \beta(\set{E})$ is bijective. The injectivity follows from the usage of a node or edge type for each attribute respectively.
    \end{proof}
    \end{lem}

\begin{lem}\label{lem_multigraph_unattrib}
Every unattributed multigraph type is at least as expressive as an integer-attributed graph type. 
Formally, $G_{ia} \preccurlyeq G_{\bar{a},m}$.

\begin{proof}
Let $g\in\set{G}_{ia}$ be an arbitrary integer-attributed graph with $g = (\set{V}, \set{E}, \alpha, \omega)$, $\set{V}\subset\mathbb{N},\ {\set{E}\subseteq\set{V}\times\set{V},}\ {\alpha\abb \set{V}\rightarrow \mathbb{N}}$ and $\omega\abb\set{E}\rightarrow\mathbb{N}$. 
Then the mapping $f\abb \set{G}_{ia}\rightarrow \text{Im}(f)\subseteq\set{G}_{\bar{a},m}$ defined as $f(g) = g'=(\set{V}', \set{E}')$, where
\begin{equation*}
\set{V}'=\bigcup\limits_{i\in[\alpha(v)]\;\abb\;  v\in\set{V}} \{\!|v|\!\},\quad 
\set{E}'=\bigcup\limits_{i\in[\omega(e)]\;\abb\;  e\in\set{E}} \{\!|e|\!\}.
\end{equation*}
is a bijective embedding.
\end{proof}
\end{lem}

Lem.~\ref{lemma:nonattributet_2_attributed}, \ref{lem_attributed_heterogen} and \ref{lem_multigraph_unattrib}, lead immediately to thm.~\ref{theorem_partially_attributed}.


\begin{rem}
Lem.~\ref{eq:hyper}, \ref{lem_directed_heterogeneous}, \ref{lemma:nonattributet_2_attributed}, \ref{lem_attributed_heterogen} and \ref{lem_multigraph_unattrib} naturally hold in the fully attributed case, where the additional attributes remain unchanged.
\end{rem}


\section{Universal Modeling Power of Attributed Graphs}\label{section_proofs_attributed}

Due to the proof of  lem.~\ref{lemma:nonattributet_2_attributed} there exists an embedding function from each unattributed graph to an attributed one, encoding the same information. Therefore, it is assumed in everything that follows that the graph types in consideration are attributed. Unless otherwise mentioned the notation is kept simple and the attributedness of the graph types is not explicitly mentioned. 
Additionally, for any two graphs $(\set{V}, \set{E}, \alpha, \omega), (\set{V}', \set{E}', \alpha', \omega')$, with node attribute sets $\set{A}, \set{A}'$ and edge attribute sets $\set{B}, \set{B}'$, it can be assumed w.l.o.g.~that it is $\set{A} = \set{A}'$ and $\set{B} = \set{B}'$ \footnote{Otherwise, the attribute domains can be extended arbitrarily.}.
$\set{C}$ will be used to refer to an extended attribute space to prevent confusion. 


\begin{lem}\label{eq:directed}
  Let $\set{G}_{d}$ be the sets of attributed directed graphs and $\set{G}_{\bar{d}}$ the set of attributed undirected graphs. Then it holds
    $\set{G}_d\approx \set{G}_{\overline{d}}$.

\begin{proof}
  Let $g\in\set{G}_{d}$ be an arbitrary directed graph of the form 
  $g = (\set{V}, \set{E}, \alpha, \omega) \text{ with } \set{V}\subset\mathbb{N},$ $\; {\set{E}\subseteq\set{V}\times\set{V}}, $ $\; {\alpha\abb\set{V}\rightarrow\set{C}},\; {\omega\abb\set{E}\rightarrow\set{C}.}$
  \begin{compactenum}
  \item[$\preccurlyeq$:] Let $f\abb \set{G}_{d}\rightarrow \text{Im}(f)\subseteq\set{G}_{\bar{d}}$ be a mapping defined by $f(g) = g'=(\set{V}',\set{E}',\alpha',\omega')$, where
    \begin{compactenum}
    \item $\set{V}' = \set{V}$, \quad$\set{E}' = \{\{u,v\}\mid (u,v)\vee (v,u)\in\set{E}\}$,\quad $\alpha'(v) \mapsto \alpha(v)\ \forall v\in \set{V}'$ and
    \item $\omega'(\{u,v\}) \mapsto  
      \begin{cases}\left(\omega((u,v)), 1\right), & \text{if } (u,v)\in\set{E}\wedge (v,u)\notin\set{E},\\
        \left(\omega((v,u)), -1\right), & \text{if } (v,u)\in\set{E}\wedge (u,v)\notin\set{E},\\
        \bigl(\left(\omega((u,v)), 1\right), \left(\omega((v,u)), -1\right)\bigr) & \text{if } (u,v),(v,u)\in\set{E}.\\
      \end{cases}$
    \end{compactenum}
    Then $f$ is  bijective.
    \item[$\succcurlyeq$:] This direction is straight forward, since every undirected edge formally corresponds to two directed edges. 
  \end{compactenum}
\end{proof}
\end{lem}


\begin{lem}\label{eq:multi}
  The set of attributed multigraphs and attributed non-multigraphs are equally expressive, i.e.,
    $\set{G}_{m}\approx \set{G}_{\bar{m}}$.

  \begin{proof}
    Let $g\in\set{G}_{m}$ be an arbitrary multigraph of the form $g = (\set{V}, \set{E}, \alpha, \omega)$ with $\set{V}\subset\{\!|v\mid v\in\mathbb{N}|\!\},$ $\; {\set{E}\subseteq\{\!| (u,v)\mid u,v\in\set{V}|\!\},}\; $ ${\alpha\abb\set{V}\rightarrow\set{C},\; \omega\abb\set{E}\rightarrow\set{C}}$.
    \begin{compactenum}
    \item[$\preccurlyeq$:] Let $f\abb \set{G}_{m}\rightarrow \text{Im}(f)\subseteq\set{G}_{\bar{m}}$ be a mapping with $f(g)=g'=(\set{V}',\set{E}',\alpha',\omega')$ and $\set{V}' = \{v\mid v\in\set{V}\}$,\quad $\set{E}' = \{e\mid e\in\set{E}\}$, \quad $\alpha'(v)\mapsto (\alpha(v)\mid v \in \set{V})$ and  $\omega'(e)\mapsto (\omega(e)\mid e \in \set{E})$.   
    Since the multisets are ordered the function is injective.
    \item[$\succcurlyeq$:] This direction is straight forward since every set is a multiset.
    \end{compactenum}
  \end{proof}
\end{lem}


\begin{lem}\label{eq:heterogeneous}
The set of attributed heterogeneous and attributed homogeneous graphs are equally expressive, i.e.,
    $\set{G}_{het}\approx \set{G}_{\overline{het}}$.

   \vspace*{-8pt}\begin{proof}
    Let $g\in\set{G}_{het}$ be an arbitrary heterogeneous graph of the form $g = (\set{V}, \set{E}, \alpha, \omega)$ with ${\set{V}\subset\mathbb{N}\times\set{S}},$ $\; {\set{E}\subseteq\{\set{V},\set{V}\}\times\set{R}},$ $ {\alpha\abb\set{V}\rightarrow\set{C},}$ $\; \omega\abb\set{E}\rightarrow\set{C}$.
    \begin{compactenum}
    \item[$\preccurlyeq$:] Let $f\abb \set{G}_{het}\rightarrow \text{Im}(f)\subseteq\set{G}_{\overline{het}}$ be a mapping with $f(g) = g'=(\set{V}',\set{E}',\alpha',\omega')$, where
      \begin{compactenum}
      \item $\set{V}' = \{v\mid (v,s)\in\set{V}\}$, \quad$\set{E}' = \{(u,v)\mid (u,v,r)\in\set{E}\}$,
      \item $\alpha'(v) \mapsto (\alpha((v,s)),s\mid (v,s)\in \set{V})\ \forall v\in \set{V}'$, and $\omega'(e) \mapsto (\omega((e,r)),r\mid (e,r)\in \set{E})\ \forall e\in \set{E}'$.
      \end{compactenum}
      The concatenation of the node and edge types to the attributes directly constitutes the injectivity of the mapping.
    \item[$\succcurlyeq$:] This direction is straight forward by introducing one node and one edge type since every homogeneous graph is a heterogeneous graph with just one node and edge type.
    \end{compactenum}
  \end{proof}
\end{lem}


\begin{lem}\label{lem_dynamic_static}
  The set of attributed dynamic and attributed static graphs are equally expressive, i.e.,
${\set{G}_{dy}\approx \set{G}_{\overline{dy}}}$.

   \vspace{-8pt}\begin{proof}
   Let $G\in\set{G}_{dy}$
     \,be an arbitrary dynamic graph of the form $G = (\set{V}_t, \set{E}_t, \alpha_t, \omega_t)_{t\in[T]} $ with ${\set{V}_t\subset\mathbb{N}},$ $\; {\set{E}_t\subseteq\set{V}_t\times\set{V}_t,}$ ${ \alpha_t\abb\set{V}_t\rightarrow\set{C},\; \omega_t\abb\set{E}_t\rightarrow\set{C}\quad \forall t\in [T]}$.
    \begin{compactenum}
    \item[$\preccurlyeq$:] Let $f\abb \set{G}_{dy}\rightarrow \text{Im}(f)\subseteq\set{G}_{\overline{dy}}$ be a mapping with $f(G) = g'=(\set{V}',\set{E}',\alpha',\omega')$, where
      \begin{compactenum}
      \item $\set{V}' = \bigcup\limits_{t\in [T]}\set{V}_t $, \quad$\set{E}' =  \bigcup\limits_{t\in [T]}\set{E}_t$,
      \item $\alpha'(v) \mapsto ((\alpha_t(v),t)\mid v\in \set{V}_t)\ \forall v\in \set{V}'$, and $\omega'(e) \mapsto ((\omega_t(e),t)\mid e\in \set{E}_t)\ \forall e\in \set{E}'$.
      \end{compactenum}
      Despite the union of the node and edge sets, the function is injective due to the ordering of the graph attributes in the static graph.
    \item[$\succcurlyeq$:] 
    This direction is straight forward by introducing one time stamp since every static graph is dynamic with only one time stamp.
    \end{compactenum}
  \end{proof}
\end{lem}

\paragraph{\textbf{Proof of Theorem \ref{theorem:attributesGraphTYpesEquallyExpressive}}}
  All attributed graph types are equally expressive.
  \begin{proof}
  \vspace*{-8pt}With lem.~\ref{lem_dynamic_static} dynamic and static graph types are equally expressive. In all other lemmata we did not distinguish between dynamic and static graphs. Therefore, by ring closure these types are equally expressive to the hypergraph type (cf.~lem.~\ref{eq:hyper}), the multigraph type (cf.~lem.~\ref{lem_multigraph_unattrib}), the (un)directed type (cf.~lem.~\ref{eq:directed}), and the heterogeneous type (cf.~lem.~\ref{eq:heterogeneous}).
    \begin{align*}
      &\underset{dynamic}{\underbrace{ \set{G}_{h} \overset{lem.~\ref{eq:hyper}}{\approx} \set{G}_{\overline{d}}  \overset{lem.~\ref{lem_multigraph_unattrib}}{\approx} 
        \set{G}_{m} \overset{lem.~\ref{eq:directed}}{\approx}  \set{G}_{d}  \overset{lem.~\ref{eq:heterogeneous}}{\approx} \set{G}_{het} }}  \\
      &\hspace*{3.1cm}\overset{lem.~\ref{lem_dynamic_static}}{\approx} \\
      &\overset{static}{\overbrace{ \set{G}_{h} \overset{lem.~\ref{eq:hyper}}{\approx} \set{G}_{\overline{d}}  \overset{lem.~\ref{lem_multigraph_unattrib}}{\approx} 
        \set{G}_{m} \overset{lem.~\ref{eq:directed}}{\approx}  \set{G}_{d}  \overset{lem.~\ref{eq:heterogeneous}}{\approx} \set{G}_{het} }}.
    \end{align*}
    
    \vspace{-12pt}
  \end{proof}
    

\section{Expressivity of Composed Graph Types}\label{sec_expressivity_composed_graphs}
In sec.~\ref{section_proofs_unattributed}, \ref{section_proofs_partially_attributed}, and \ref{section_proofs_attributed} we investigated the expressivity relation between graph types defined by a single graph property. However, as seen in def.~\ref{defn:StaticGraphsComb}, composed graph types can also be related to each other by their expressivity relation.
In this case, a bijective embedding between multiple composed graphs consists of the composition of the individual embeddings determined for single structural properties from sec.~\ref{section_proofs_unattributed}, \ref{section_proofs_partially_attributed}, and \ref{section_proofs_attributed}. The composed embedding is illustrated in fig.~\ref{fig:trafo_of_combined_graphs}.

Let $\set{G}_{1}$ and $\set{G}_2$ be two arbitrary (composed) graph types and $p_i$ for all indices $i \in [1, \ldots, j]$  the structural graph properties or a property describing an elementary graph from def.~\ref{defn:StaticGraphsElm} and \ref{defn:StaticStructuralProp}. Here, $j$ is the number of properties that have to be changed to come from $\set{G}_1$ to $\set{G}_2$. Further, let $t(g, p_i)$ be the embedding for adding or removing the property $p_i$ to or from a graph $g$.
Then, the following composition of single embeddings $t(g, p_i)$ describes one existent bijective embedding from graph $g \in \set{G}_{1}$ into the graph type $\set{G}_{2}$, namely
\begin{equation*}
t(g, [p_1, p_2,\ldots, p_j]) = t(g, p_1) \circ t(g, p_2) \circ \cdots \circ t(g, p_j) \in \set{G}_2.
\end{equation*}

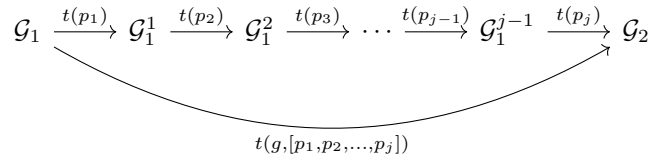
\begin{figure}[h!]
\begin{center}
\begin{tikzcd}
\set{G}_1\rar{t(p_1)}\arrow[black, bend right]{rrrrr}[black,swap]{t(g, [p_1, p_2,\ldots, p_j])}  
& \set{G}_1^1 \rar{t(p_2)}  
& \set{G}_1^2 \rar{t(p_3)}
& \cdots \rar{t(p_{j-1})}
& \set{G}_1^{j-1} \rar{t(p_{j})}    
& \set{G}_2
\end{tikzcd}
\end{center}
\caption{\label{fig:trafo_of_combined_graphs}embeddings between composed graphs as composition of single structural property embeddings.}
\end{figure}
Note that, having fixed the embedding functions for simple graph type embeddings, the order within the composition of embeddings is not unique. I.e. in particular, that the embedding from $\set{G}_1$ to $\set{G}_2$ is deterministic in the result but not in the process. 

%% file: chapters/conclusion.tex
\section{Conclusion}\label{sec:conclusion}
In order to clarify the modeling power of graphs, we evaluated the power to encode different properties, i.e. the expressivity power, for common graph types.
First, we compared unattributed graph types w.r.t.~their expressivity power. We proved that hyper- and non-hypergraphs are equally expressive, i.e., $\set{G}_{h}\approx \set{G}_{\bar{h}}$ and that heterogeneous graphs are at least as expressive as directed graph types, i.e., $\set{G}_d \preccurlyeq \set{G}_{het}$. 
For this purpose, we showed the existence of bijective embeddings between the graph types. Next, in the case of one attributed and one unattributed graph type, we could show that it holds $\set{G}_{\bar{a}} \preccurlyeq \set{G}_a$ for unattributed and attributed graphs. 
In contrast, for attributed and unattributed heterogeneous graphs and integer-attributed and unattributed multigraphs, it is $\set{G}_{a} \approx \set{G}_{\bar{a}, het}$, and ${\set{G}_{ia} \approx \set{G}_{\bar{a}, m}}$, respectively.
 Subsequently, we proved that attributed graphs of any graph type could encode any graph property leading to the equal expressivity power of attributed graph types, cf.~thm.~\ref{theorem:attributesGraphTYpesEquallyExpressive}. 
 Finally, we noted the expressivity relation between composed graphs resulting from the composition of existing embeddings between the attributed and unattributed graph types having just one graph property.

%% file: chapters/acknowledgement.tex
\section*{Acknowledgements}
This work was supported by the Ministry of Education and Research Germany (BMBF, grant number 01IS20047A). Further, we would like to thank Abdul Hannan, Franz Götz-Hahn, Jan Schneegans and Bernhard Sick for review of the manuscript and fruitful discussions.


%% file: chapters/arxiv/bibliography.tex
\bibliography{bibliography/bibliography}